\documentclass[11pt]{article}
\usepackage{times}
\usepackage{latexsym}
\usepackage{epsfig,enumerate,amsmath,amsfonts,amssymb,setspace}
\usepackage{indentfirst}
\usepackage{wrapfig,lipsum}
\usepackage{epsfig}
\usepackage{graphicx}
\usepackage{amsthm}
\usepackage{algorithm}
\usepackage{algorithmic}
\usepackage[T1]{fontenc}
\usepackage{authblk}
\usepackage[algo2e,ruled,vlined]{algorithm2e}
\newtheorem{definition}{Definition}[section]
\newtheorem{theorem}{Theorem}[section]
\newtheorem{lemma}{Lemma}[section]

\newtheorem{claim}{Claim}[section]

\title{Complexity of Paired Domination in AT-free and Planar Graphs}
\author[1]{Vikash Tripathi\thanks{2017maz0005@iitrpr.ac.in}}
\author[2]{Ton Kloks}
\author[1]{Arti Pandey\thanks{arti@iitrpr.ac.in}}
\author[1]{Kaustav Paul\thanks{kaustav.20maz0010@iitrpr.ac.in}}
\author[3]{Hung-Lung Wang \thanks{anil@scs.carleton.ca}}
\affil[1]{Department of Mathematics, Indian Institute of Technology Ropar, Punjab, India.}
\affil[2]{klokston@gmail.com}
\affil[3]{Department of Computer Science and Information Engineering, National Taiwan Normal University, Taiwan}

\date{}
\begin{document}
\maketitle
\begin{abstract}
For a graph $G=(V,E)$, a subset $D$ of vertex set $V$, is a dominating set of $G$ if every vertex not in $D$ is adjacent to atleast one vertex of $D$. A dominating set $D$ of a graph $G$ with no isolated vertices is called a \emph{paired dominating set (PD-set)}, if $G[D]$, the subgraph induced by $D$ in $G$ has a perfect matching. The \textsc{Min-PD} problem requires to compute a PD-set of minimum cardinality. The decision version of the \textsc{Min-PD} problem remains NP-complete even when $G$ belongs to restricted graph classes such as bipartite graphs, chordal graphs etc.
On the positive side, the problem is efficiently solvable for many graph classes including intervals graphs, strongly chordal graphs, permutation graphs etc. In this paper, we study the complexity of the problem in AT-free graphs and planar graph. The class of AT-free graphs contains cocomparability graphs, permutation graphs, trapezoid graphs, and   interval graphs as subclasses. We propose a polynomial-time algorithm to compute a minimum PD-set in AT-free graphs. In addition, we also present a linear-time $2$-approximation algorithm for the problem in AT-free graphs. Further, we prove that the decision version of the problem is NP-complete for planar graphs, which answers an open question asked by Lin et al. (in Theor. Comput. Sci., $591 (2015): 99-105$ and Algorithmica, $ 82 (2020) :2809-2840$).\\

\end{abstract}

\section{Introduction}\label{sec1}

Let $G=(V,E)$ be a graph. A vertex $v \in V$ is \emph{adjacent} to another vertex $u \in V$ if $uv$ is an edge of $G$. In this case, we say $u$, a neighbour of $v$. The set of all vertices adjacent to $v \in V$, denoted by $N_{G}(v)$, is known as \emph{open neighbourhood} of $v$, whereas the set $N_{G}[v]= N_{G}(v) \cup \{v\}$ is known as \emph{closed neighbourhood} of $v$ in $G$. 

In a graph $G=(V,E)$, a vertex $v \in V$ \emph{dominates} a vertex $u \in V$ if $u \in N_G[v]$. A subset $D$ of vertex set $V$, is a \emph{dominating set} of $G$ if every vertex of $V$ is dominated by at least one vertex of $D$. The \emph{domination number}, symbolized as $\gamma(G)$, is the minimum cardinality of a dominating set. The concept of domination has wide applications and is thoroughly studied by researchers in the literature. A survey of the results, both algorithmic as well as combinatorial, on domination can be found in \cite{hhs1,hhs2}.  Due to several applications in the real world problems, numerous variations of domination are introduced by imposing one or more additional condition on dominating set. Many of these variations are thoroughly studied by researchers in the literature. Total domination is one of the important variation of domination. For a graph $G=(V,E)$ without an isolated vertex, a \emph{total dominating set} of $G$ is a subset $D$ of vertex set such that every vertex of the graph is  adjacent to at least one vertex in $D$.


Paired domination is another important variation of domination, introduced by Haynes and Slater in \cite{paired}. A detailed survey of results on domination problem and its variations can also be found in a recent book by Haynes et al.~\cite{hhs3}. Given a graph $G=(V,E)$ with no isolated vertices, a subset $D$ of vertex set $V$, is a \emph{paired dominating set(PD-set)} if $D$ is a dominating set and the subgraph induced by $D$ in $G$ has a perfect matching. The \emph{paired domination number}, symbolized as $\gamma_{pr}(G)$, is the cardinality of a minimum PD-set of $G$. The \textsc{Min-PD} problem requires to compute a PD-set of a graph $G$ without an isolated vertex. More precisely, the \textsc{Min-PD} problem and its decision version of the same are defined as follows:

\medskip

\noindent\underline{\textsc{Min-PD} problem}
\\
[-12pt]
\begin{enumerate}
  \item[] \textbf{Instance}: A graph $G$ with no isolated vertices.
  \item[] \textbf{Solution}: A PD-set $D$.
  \item[] \textbf{Measure}: Size of $D$.
\end{enumerate}

\noindent\underline{\textsc{Decide PD-set} problem}
\\
[-12pt]
\begin{enumerate}
  \item[] \textbf{Instance}: A graph $G$ and an integer $k>0$, satisfying $k \le \vert V \vert$.
  \item[] \textbf{Query}: Is there is a PD-set $D$ of $G$, satisfying $\vert D\vert  \le k$?
\end{enumerate}

It is shown that the decision version of the problem is NP-complete for general graphs~\cite{paired}. Therefore, complexity of the problem is studied for several restricted graph classes. It is proven that, the decision version of the problem is NP-complete when restricted to special graph classes, including bipartite graphs~\cite{paired3}, perfect elimination bipartite graphs~\cite{paired5}, and split graphs~\cite{paired3}. But, on the good side, the problem is efficiently solvable in several important graph classes, including  permutation graphs~\cite{paired7}, interval graphs~\cite{paired3}, block graphs~\cite{paired3}, strongly chordal graphs~\cite{paired10},  circular-arc graphs~\cite{paired6} and some others. A detailed survey of the results on paired domination can be found in~\cite{paired4}. In Fig.~\ref{fig:1} we show the hierarchy of some important graph classes and the complexity status of the \textsc{Decide PD-set} problem in these graph classes.

\begin{figure}[h]
\begin{center}
\includegraphics[width=0.9 \textwidth]{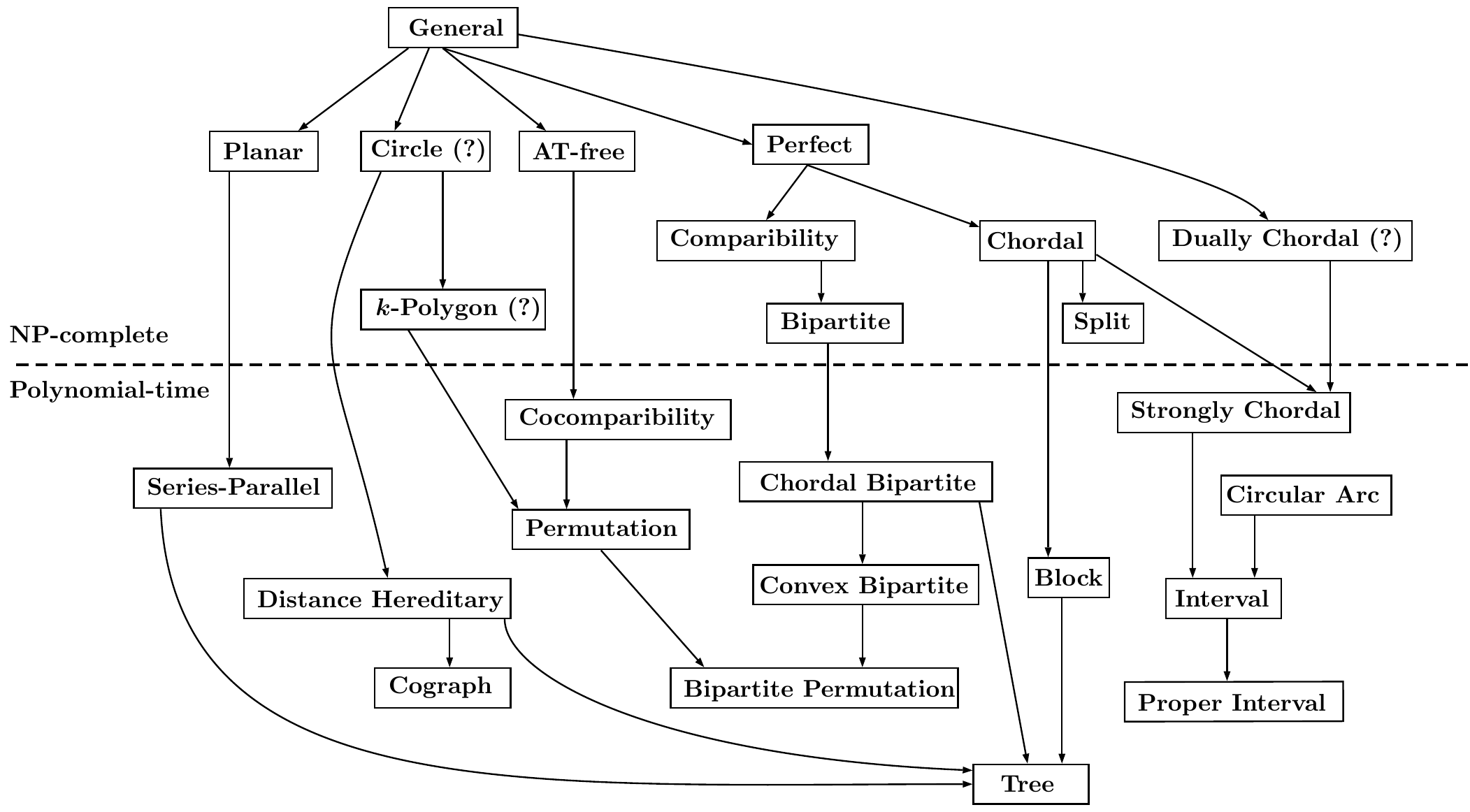}
\caption{Complexity status of \textsc{Min-PD} problem in some well known graph classes.}
\label{fig:1}
\end{center}
\end{figure}

The computational complexity of the problem is still unknown in some graph classes including planar graphs, AT-free graphs and circle graphs. AT-free graphs is introduced by Corneil et al. in \cite{at}. AT-free graph class includes some important classes of graphs such as interval graphs, permutation graphs and cocomparability graphs as subclasses. A minimum dominating and total dominating set of an AT-free graph can be computed in polynomial-time, see \cite{at1}. In this paper, we investigate the computational complexity of the problem on AT-free graph and planar graphs. We show that minimum PD-set of an AT-free graph can be computed in polynomial-time. In addition, we give an approximation algorithm which computes a PD-set of any AT-free graph, within a factor of $2$. Lin et al. in \cite{paired6} and \cite{pddh} asked to determine the complexity of the problem in planar graphs. In this paper, we prove that \textsc{Decide PD-set} problem remain NP-complete even for  planar graphs. The section wise contribution of the paper is outlined as follows:

In Section~\ref{sec2}, we give insights on some notations and definitions, including properties of AT-free graphs. In Section~\ref{sec3}, we prove the existence of a linear-time $2$-approximation algorithm to compute a PD-set of an AT-free graph. In Section~\ref{sec4}, we design a polynomial time algorithm to compute a minimum cardinality  PD-set of an AT-free graph. In Section~\ref{sec5}, we show that the problem remains NP-hard for  planar graphs. Finally, Section~\ref{sec6} wind up the paper with some interesting open questions on the problem.

\section{Preliminaries}
\label{sec2}

\subsection{Basic Notations and Definitions}

In this paper, we consider only simple, connected and finite graphs with no isolated vertices. Let $G=(V,E)$ be a graph. The sets $V(G)$ and $E(G)$ represents node(vertex) set and edge set respectively of the graph. When there is no ambiguity regarding graph $G$, for simplification, we use $V$ and $E$ to denote of $V(G)$ and $E(G)$ respectively. For an edge $e = uv \in E$, $u$ and $v$ are called \emph{end vertices} of $e$. For any non-empty set $A \subseteq V$, the \emph{open neighbourhood of A}, symbolized as $N_{G}(A)$, is given by $N_{G}(A) = \bigcup_{v \in A} N_{G}(v)$ whereas the set $N_{G}[A] = N_{G}(A) \cup A$ is known as \emph{closed neighbourhood of A}. Further, for a set $A \subseteq V$, $G \setminus A$ represents the graph obtained by deleting vertices of set $A$ and all edges having at least one end vertex in $A$, from the graph. In case, $A=\{u\}$, we use $G \setminus u$, instead of using $G \setminus \{u\}$.

A subset $X$ of vertex set is an \emph{independent set} if no two vertices of $X$ are adjacent in $G$. A path $P$ in $G$ is a sequence of vertices $(x_1,x_2, \ldots, x_n)$ such that $(x_i,x_{i+1}) \in E$ for each $i \in \{1,2, \ldots, n-1\}$. For a path  $P =(x_1,x_2, \ldots, x_{n+1})$ in $G$, the length of $P$ is $\vert V(P)-1\vert =n$. Let $x,y \in V(G)$. The distance between $x$ and $y$ in the graph $G$, denoted by $d_{G}(x,y)$, is the length of a shortest path between $x$ and $y$. The diameter of a graph $G$, denoted by $diam(G)$,  is defined as $diam(G) = $ max$\{d_{G}(x,y) \mid x,y \in V(G)\}$. We use the standard notation $[n]$ to denote the set $\{1,2,\ldots,n\}$.

\subsection{AT-free Graphs}

\emph{Let $G=(V,E)$ be a graph}. A set $T= \{p,q,r\}$ of three vertices, is called an \emph{asteroidal tripe}(in short AT) if $T$ is an independent set and for any two vertices in the set $T$ there exits a path $\mathcal{P}$ between them such that $V(\mathcal{P})$ does not contain any vertex from the closed neighbourhood of third.
A graph is \emph{AT free} if it does not contain an asteroidal tripe. A path on six vertices is an example of an AT-free graph.

\begin{definition}
In a graph $G=(V,E)$, a pair of vertices $(x,y)$ is called a dominating pair, if the vertex set of any path between $x$ and $y$ in $G$ is a dominating set of $G$. A dominating shortest path is a shortest path connecting $x$ and $y$ in $G$. 
\end{definition}

\begin{figure}[h]
\begin{center}
\includegraphics[width=0.45\textwidth]{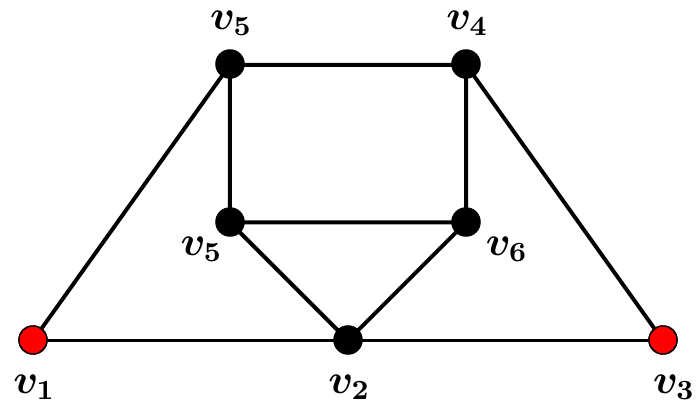}
\caption{An AT-free graph $G$}
\label{fig::3}
\end{center}
\end{figure}

\noindent An asteroidal triple free graph is shown in Fig.~\ref{fig::3}. For the graph $G$ in Fig.~\ref{fig::3}, $(v_{1},v_{3})$ is a dominating pair, and $P=(v_{1},v_{2},v_{3})$ is a dominating shortest path. We have the following result for a connected AT-free graph in the literature.

\begin{theorem}~\cite{at,at3}
\label{th:1}
A dominating pair exists in every AT-free graph which can be computed in linear time.
\end{theorem}

\section{Approximation Algorithm}
\label{sec3}

In this section, we show that a PD-set of an AT-free graph $G$, can be computed in linear time whose cardinality is at most twice of $\gamma_{pr}(G)$.
Let $G$ is an AT-free graph. Using Theorem~\ref{th:1}, we note that there exists a dominating pair $(x,y)$ in $G$. Assume that $P$ is a dominating shortest path between $x$ and $y$ in $G$, and the number of vertices in $P$ are $t$. Note that any vertex that is not in $P$ is adjacent to some vertex of $V(P)$, as the set $V(P)$ is a dominating set of $G$. We may also conclude that any vertex not in $P$ has at most three neighbours in $P$, since otherwise $P$ will not be a shortest path. By a similar argument we note that any two adjacent vertices in $G$ dominate at most the vertices of a $P_4$ in $P$. Consequently, $\frac{\gamma_{pr}}{2} \geq \lceil \frac{t}{4} \rceil$, that is, $\gamma_{pr} \geq 2 \cdot \lceil \frac{t}{4} \rceil$. Before proving the Theorem~\ref{thapx}, which is the main result of this section, we notice that the following lemma is true.

\begin{lemma}
\label{lemma:b}
For any odd positive integer $n$, $\lceil \frac{n}{4} \rceil \geq \frac{n+1}{4}$.
\end{lemma}
\begin{proof}
The proof is easy, and hence is omitted.
\end{proof}

\begin{theorem}
\label{thapx}
Given an AT-free graph $G$, a PD-set $D$ of $G$ can be computed in linear time, satisfying $\vert D \vert \le 2 \cdot \gamma_{pr}(G)$.
\end{theorem}

\begin{proof}
Given an AT-free graph $G=(V,E)$, there is a linear-time algorithm to find a dominating pair $(x,y)$ of $G$ (by Theorem~\ref{th:1}). Let $P=(x=v_1,v_2\ldots v_{t-1}, v_t=y)$ be a shortest path between $x$ and $y$, and $D = V(P)$.  We have already observed that $\gamma_{pr}(G) \geq  2 \cdot \lceil \frac{t}{4} \rceil$. We prove the result under the following assumptions:

\medskip
\noindent \textbf{Case 1:} If $t$ is even.\\
Here, we note that the set $D$ is a PD-set and $\vert D\vert  = t \leq 4 \cdot \lceil \frac{t}{4} \rceil \leq 2 \cdot \gamma_{pr}(G)$.

\medskip
\noindent \textbf{Case 2:} If $t$ is odd.\\
In this case, we construct a PD-set of the graph $G$ by adding at most one vertex in $D$. Clearly, $D$ is a dominating set. For pairing, we pair $v_i$ with $v_{i+1}$ for $i \in [t-2]$. Now we need to pair $v_t$. Note that if $N(v_t) \subseteq D$ then $D \setminus \{v_t\}$ is a PD-set of $G$, otherwise if there exists a vertex $u \in N(v_{t}) \setminus D$ then the updated set $D=D \cup \{u\}$ is a PD-set of $g$. Therefore, we can always construct a PD-set $D$ of $G$, where  $\vert D\vert  \leq t+1$. Using Lemma~\ref{lemma:b}, we have $t+1 \leq 4 \cdot \lceil \frac{t}{4} \rceil $. Hence, $\vert D\vert \leq t+1 \leq 2 \gamma_{pr}(G)$.

\medskip
In both the cases, we can obtain a PD-set $D$ satisfying, $\vert D\vert  \leq 2 \gamma_{pr}(G)$. Hence, we have an efficient $2$-approximation algorithm to computes a PD-set of an AT-free graph.
\end{proof}

\section{Exact Polynomial-time Algorithm}
\label{sec4}

The main purpose this section is to establish a polynomial time algorithm that outputs a minimum cardinality PD-set, when the input graph is an AT-free graph. For this, we first present a theorem, which will be useful in designing our algorithm. In this theorem, we show that there exists a BFS-tree $T$ of $G$ and a minimum PD-set $D$ of $G$ such that the number of vertices of $D$ in some consecutive levels of $T$ are bounded. We will use the notation $L_{i}$ to denote the vertices, which are at $i^{th}$ level in the tree $T$, that is, the set of vertices which are at distance $i$ from the root node in tree $T$.
The following result is already known in literature.

\begin{theorem}~\cite{kloks1}
Let $G$ be an AT-free graph with dominating pair $(x,y)$ and $T$ be a BFS-tree of $G$ rooted at $x$. Let $L_0, L_1, L_{2}, \ldots ,L_l$ are the BFS-levels of the BFS-tree $T$. Then there exists a linear-time algorithm which computes a path $P=(x=x_{0},x_{1},x_{2},\ldots,x_{d}=y)$ such that $x_i \in L_{i}$ for each $0 \le i \le d$ and every vertex $w \in L_{i}$ for $i\in \{1,2,\ldots,l\}$ is adjacent to either $x_{i-1}$ or $x_{i}$.
 \end{theorem}

\begin{theorem}\label{t:thm1}
Let $G=(V,E)$ be an AT-free graph and $(x,y)$ be a dominating pair of $G$. If $L_0, L_1, L_{2}, \ldots ,L_l$ are the BFS-levels of the BFS-tree $T$ rooted at $x$ then there exists a minimum cardinality PD-set $D_{p}$ of $G$ such that $\vert  D_p \cap \bigcup_{k=i}^{i+j} L_k \vert  \leq j+4$ for all $i \in \{0, 1, \ldots l\}$ and $j \in \{0, 1, \ldots l-i\}$.
\end{theorem}

\begin{proof}
Let $G=(V,E)$ be an AT-free graph and $D_p$ be a minimum cardinality PD-set of the graph $G$. Suppose that the set $D_p$ does not satisfy the given property, that is, there is at least one pair $(i,j)$ such that $\vert  D_p \cap \bigcup_{k=i}^{i+j} L_k \vert  > j+4$ where $i \in \{0, 1, \ldots l\}$ and $j \in \{0, 1, \ldots l-i\}$. Let $B = \{ (i,j): ~ \vert  D_p \cap \bigcup_{k=i}^{i+j} L_k \vert  \geq j+5 \}$. Note that $B \neq \emptyset$. Now we choose pair $(i', j')$ such that $i' = $ min$\{i \vert  (i,j) \in B\}$ and $j' = $ max$\{j \mid (i',j) \in B\}$. By the choice of the pair $(i',j')$, note that $D_p \cap L_{i'-1} = \emptyset$ and $D_p \cap L_{i'+j'+1} = \emptyset$. Using the properties of a BFS-tree, we note that for any vertex $v \in (D_p \cap \bigcup_{k=i'}^{i'+j'} L_k)$, any neighbor of $v$ belongs to one of the levels $L_{i'-1}, L_{i'},  \cdots L_{i'+j+1}$. Let $A=\{x_{i'-2}, x_{i'-1}, \ldots x_{i'+j'+1}\}$. Note that, $\vert A\vert  = j'+4$. Since $V(P)$ is a dominating set of $G$ and each vertex $z \in L_{i}$ is adjacent to either $x_{i-1}$ or $x_i$, $\bigcup_{k=i'-1}^{i'+j'+1} L_k \subseteq N[A]$. Now by updating $D_p$ we will find another minimum PD-set $D'_{p}$ such that $A \subseteq D'_{p}$.

\newpage
\noindent \textbf{Case 1:} If $x_{i'-2} \notin D_p$ and  $\vert A\vert $ is even.

Since $D_p \cap L_{i'+j'+1} = \emptyset$, $x_{i'+j'+1} \notin D_{p}$. If $x_{i'-2} \notin D_p$ and  $\vert A\vert $ is even then the set $D'_{p} = (D_{p} \setminus \bigcup_{k=i'}^{i'+j'} L_k) \cup A$ is a PD-set of $G$ with $\vert D'_{p}\vert  < \vert D_{p}\vert $, a contradiction to the choice of $D_p$.

\medskip
\noindent \textbf{Case 2:} If $x_{i'-2} \notin D_p$ and  $\vert A\vert $ is odd.

Note that $\vert A\vert $ is odd and $G[A]$ is a path, if we include $A$ in a PD-set  we can pair all the vertices in $A$ except one. We pair $(x_{i'-2},x_{i'-1})$, $(x_{i'},x_{i'+1}), \ldots ,(x_{i'+j'-1},x_{i'+j'})$. Now we need to pair $x_{i'+j'+1}$. If $(N_G(x_{i'+j'+1}) \setminus \{x_{i'+j'}\}) \subseteq D_p$ and $(N_G(x_{i'+j'+1}) \setminus \{x_{i'+j'}\}) \cap (L_{x_{i'+j'+1}} \cup L_{x_{i'+j'}}) = \emptyset$. In this case using the property of path $P$, note that all the vertices in $L_{i'+j'+1}$ is adjacent to $x_{i'+j'}$. Hence the set $D'_{p} = (D_{p} \setminus \bigcup_{k=i'}^{i'+j'} L_k) \cup (A \setminus \{x_{i'+j'+1}\})$ is a PD-set of $G$ with $\vert D'_{p}\vert  < \vert D_{p}\vert $, a contradiction. If there is a vertex $u \in N_G(x_{i'+j'+1}) \setminus \{x_{i'+j'}\}$ such that $u \notin D_p$ or if $N_G(x_{i'+j'+1}) \setminus \{x_{i'+j'}\} \subseteq D_p$ but there is a vertex $u \in N_G(x_{i'+j'+1}) \setminus \{x_{i'+j'}\}$ such that $u \in (L_{x_{i'+j'+1}} \cup L_{x_{i'+j'}})$ then take $A' = A \cup \{u\}$. Note that the set $D'_{p} = (D_{p} \setminus \bigcup_{k=i'}^{i'+j'} L_k) \cup A'$ is a PD-set of $G$, implying that $\vert D'_{p}\vert  \geq \vert D_{p}\vert $. Also we have $\vert  D_p \cap \bigcup_{k=i'}^{i'+j'} L_k \vert  \geq j'+5$ and $\vert A'\vert  = j'+5$ implying that $\vert D'_{p}\vert  \leq \vert D_{p}\vert $. Hence $D'_{p}$ is also a minimum PD-set of $G$.

\medskip
\noindent \textbf{Case 3:} If $x_{i'-2} \in D_p$ and  $\vert A\vert $ is even.

Since $D_p \cap L_{i'-1} = \emptyset$ and $D_p \cap L_{i'+j'+1} = \emptyset$, no vertex in $D_{p} \setminus (\bigcup_{k=i'-1}^{i'+j'+1} L_k)$  is paired with a vertex in $\bigcup_{k=i'-1}^{i'+j'+1} L_k$ implying that $x_{i'-2}$ is not paired with any vertex of $A$. Note that $ \vert  A \setminus \{ x_{i'-2}\} \vert $ is odd. Hence, we can pair the vertices of $A \setminus \{ x_{i'-2}\}$ except one. Since $G[A \setminus \{ x_{i'-2}\}]$ is a path, we pair $(x_{i'-1},x_{i'})$, $(x_{i'+1},x_{i'+2}), \ldots$, $(x_{i'+j'-1},x_{i'+j'})$. Now we need to find a pair of $x_{i'+j'+1}$. Similar to the previous case, if $(N_G(x_{i'+j'+1}) \setminus \{x_{i'+j'}\}) \subseteq D_p$ and $(N_G(x_{i'+j'+1}) \setminus \{x_{i'+j'}\}) \cap (L_{x_{i'+j'+1}} \cup L_{x_{i'+j'}}) = \emptyset$, then using the property of path $P$, we nay observe that all the vertices in $L_{i'+j'+1}$ is adjacent to $x_{i'+j'}$. Hence the set $D_{p'} = (D_{p} \setminus \bigcup_{k=i'}^{i'+j'} L_k) \cup (A \setminus \{x_{i'+j'+1})\}$ is a PD-set of $G$ with $\vert D_{p'}\vert  < \vert D_{p}\vert $, a contradiction. If there is a vertex $u \in N_G(x_{i'+j'+1}) \setminus \{x_{i'+j'}\}$ such that $u \notin D_p$ or if $N_G(x_{i'+j'+1}) \setminus \{x_{i'+j'}\} \subseteq D_p$ but there is vertex $u \in N_G(x_{i'+j'+1}) \setminus \{x_{i'+j'}\}$ such that $u \in (L_{x_{i'+j'+1}} \cup L_{x_{i'+j'}})$ then take $A' = A \cup \{u\}$. Note that the set $D_{p'} = (D_{p} \setminus \bigcup_{k=i'}^{i'+j'} L_k) \cup A'$ is a PD-set of $G$, implying that $\vert D_{p'}\vert  \geq \vert D_{p}\vert $. Also we have $\vert  D_p \cap \bigcup_{k=i'}^{i'+j'} L_k \vert  \geq j'+5$ and $\vert A'\vert  = j'+5$ implying that $\vert D_{p'}\vert  \leq \vert D_{p}\vert $. Hence $D_{p'}$ is also a minimum PD-set of $G$.

\medskip
\noindent \textbf{Case 4:} If $x_{i'-2} \in D_p$ and  $\vert A\vert $ is odd.

Since $x_{i'-2} \notin D_p$ and $\vert A\vert $ is odd, similar to previous cases we can pair the vertices $(x_{i'-1},x_{i'})$, $(x_{i'+1},x_{i'+2}), \ldots$, $(x_{i'+j'-1},x_{i'+j'})$ and we can find a PD-set of smaller cardinality or a vertex $u$ such that if we take $A' = A \cup \{u\}$ then the set $D_{p'} = (D_{p} \setminus \bigcup_{k=i'}^{i'+j'} L_k) \cup A'$ is a PD-set of $G$, implying that $\vert D_{p'}\vert  \geq \vert D_{p}\vert $. Also we have $\vert  D_p \cap \bigcup_{k=i'}^{i'+j'} L_k \vert  \geq j' +5$ and $\vert A'\vert  = j'+ 5$ implying that $\vert D_{p'}\vert  \leq \vert D_{p}\vert $. Hence $D_{p'}$ is also a minimum PD-set of $G$.

\medskip
Further, note that if $i'$ is $0$ or $1$, we can choose $A=\{x_0,x_1, \ldots, x_{i'+j'+1}\}$ if $\vert \{x_0,x_1, \ldots, x_{i'+j'+1}\}\vert $ is even, otherwise we can choose $A=\{x_0,x_1, \ldots, x_{i'+j'+1},u\}$, where $u \in N(x_{i'+j'+1})$. We can show the existence of $u$ as we did above. In both the cases $\vert A\vert  \leq j' + 4$ implying that $D'_{p} = (D_{p} \setminus \bigcup_{k=i'}^{i'+j'} L_k) \cup A$ is a PD-set of $G$ having cardinality less than the minimum cardinality PD-set $D_{p}$ of $G$, a contradiction. Hence $i' \notin \{0,1\}$. Similarly we can claim that $i'+j' \notin \{l-1,l\}$.

\medskip
We call this replacement of $D_{p}$ with $D_p'$ an exchange step. Now, if $\vert  D_p' \cap \bigcup_{k=i}^{i+j} L_k \vert  \leq j+4$ for all $i \in \{0, 1, \ldots l\}$ and $j \in \{0, 1, \ldots l-i\}$ then $G$ has a minimum paired dominating $D'_{p}$ satisfying the condition given in Theorem~\ref{t:thm1}. Otherwise, let $B' = \{ (i,j): ~\vert  D'_{p} \cap \bigcup_{k=i}^{i+j} L_k \vert  \geq j+5 \}$. Suppose $(i,j) \in B'$. Now we will show that $i > i'$. By contradiction suppose, $i \leq i'$. In this case note that $i+j \geq i'-2$ otherwise, $(i,j) \in B$, contradicting the choice of $i'$. Also, $ \vert  D'_{p} \cap L_{t}\vert  \geq  1$ for all $t \in \{i'-2, i'-1, \ldots , i'+j'+1\}$. Hence for $(i,j) \in B'$ with $i < i'$ and $i+j \geq i'-2$ there exits a $j'$ such that $(i,j') \in B'$ and $i+j' \geq i'+j'+1$. By construction of $D'_{p}$, we note that $ \vert  D_{p} \cap \bigcup_{k=i}^{i+j'}L_k \vert  \geq \vert  D'_{p} \cap \bigcup_{k=i}^{i+j'}L_k \vert  \geq j'+5$ implying that $(i,j') \in B$, a contradiction to the choice of $i'$ or $j'$. Hence $i > i'$. Therefore, if $i^{''} = $ min $\{i \mid (i,j) \in B'\}$ then $i^{''} > i'$.

This implies that, at every exchange step, we replace a minimum cardinality PD-set $D_p$ with an updated minimum cardinality PD-set $D_{p}'$. After each exchange step, we note that the smallest value of $i$ for which there was a $j \in \{0,1, \ldots, l-i\}$ satisfying $ \vert D_{p} \cap \bigcup_{k=i}^{i+j}L_k \vert \geq j+5$, for the minimum cardinality PD-set $D_{p}$, will increase. Therefore, we conclude that, if we start with any minimum cardinality PD-set $D_{p}$, we obtain a minimum cardinality PD-set $D_{p}'$, such that $\vert  D_{p}' \cap \bigcup_{k=i}^{i+j} L_k \vert  \leq j+4$ for all $i \in \{0, 1, \ldots l\}$ and $j \in \{0, 1, \ldots l-i\}$, by executing at most $d$ exchange steps.
\end{proof}

Now we are ready to present an algorithm to compute a minimum cardinality PD-set of an AT-free $G$. Using Theorem~\ref{t:thm1}, we may conclude that there is a minimum PD-set of $G$ that contains at most $6$ vertices from any three consecutive BFS-levels of $x$, where $(x,y)$ is a dominating pair of $G$. The idea behind our algorithm is the following:

In our algorithm, we explore a BFS-level of $x$ in each iteration. In the $i^{th}$-iteration of the algorithm, we do the following:
\begin{itemize}
\item store all the possible sets $X' \subseteq \bigcup_{j=0}^{i+1}L_j$ such that $X'$ dominates all the vertices till $i^{th}$-level.
\item ensure that all the vertices in $X' \cap (\bigcup_{j=0}^{i}L_j)$ are paired as these vertices can not be paired with a vertex at level $i+2$ or above.
\item for every possible set $X'$, store another set $X = X' \cap (L_{i} \cup L_{i+1})$
\end{itemize}

The set $X$ helps in extending a partial solution $X'$ to the next level as we are restricted to select at most $6$ vertices from any three consecutive levels in a minimum PD-set. Below, we have provided the detailed algorithm for computing a minimum cardinality PD-set $D_{p}$ of an AT-free graph $G$. The set $D_{p}$ maintains the property that it contains at most $6$ vertices from any three consecutive BFS-levels of $x$.

\begin{algorithm}[ht!]
\label{algo}\small
\textbf{Input:} A connected AT-free graph $G=(V,E)$ with a dominating pair $(x,y)$; \\
\textbf{Output:} A PD-set $D_{p}$ of $G$;\\
Compute the BFS-levels of $x$;\\
For $0 \leq i \leq l$, let $L_{i} = \{w \in V \mid d_{G}(x,w) = i\}$ denote the set of vertices at level $i$ in the BFS of $G$ rooted at $u$.\\
In particular, $L_{0} = \{x\}$.\\
Initialize the queue $Q_1$ which contains an ordered tuple $(X,X,size(X))$  for all non-empty $X \subseteq N[x]$ such that $size(X) = |X| \leq 6$;\\
Initialize $i=1$;\\
\While{$(Q_i \neq \emptyset$ and $i<l)$}{
	Update $i=i+1$;\\
	\For{$($each element $(X,X',size(X'))$ of the queue $Q_{i-1})$}{
		\For{$($every $U \subseteq L_i$ with $|X \cup U | \leq 6)$}{
			\If{$(L_{i-1} \subseteq N[X \cup U]$ and there exists a set $U' \subseteq U$ such that $G[X' \cup U'] $ has a perfect matching $)$}{
			$Y=(X \cup U) \setminus L_{i-2}$;\\
			$Y'=X' \cup U$;\\
			$size(Y') = size(X') + |U|$;\\
				\If{$($for all element $(X,X',size(X'))$ of $Q_i$, $X \neq Y)$}{
				insert $(Y, Y', size(Y'))$ in the queue $Q_i$;
				}
				\If{$($there is a tuple $(Z,Z',size(Z'))$ in $Q_i$ such that 	
				$Z=Y$ and $size(Y') < size(Z'))$}{
				delete $(Z,Z',size(Z'))$ form $Q_i$;\\
				insert $(Y, Y', size(Y'))$ in $Q_i$;
				}
			}
		}
	}
}	
Among all the triples $(X,X',size(X'))$ in the queue $Q_l$ that satisfy $L_l \subseteq N[X]$ and $G[X']$ has a perfect matching, find one such that $size(X')$ is minimum, say $(D,D',size(D'))$;\\
$D_p = D'$;\\
return $D_{p}$;	
\caption{\textbf{Minimum Paired Domination in AT-free Graphs}}
\end{algorithm}

\medskip
Now we prove the following theorem to show that the Algorithm~\ref{algo} returns a minimum PD-set. We also analyse the running time of the algorithm.

\begin{theorem}
\label{t:correct}
Let $G=(V,E)$ be an AT-free graph such that $\vert V\vert =n$ and $\vert E \vert = m$. Algorithm~\ref{algo} computes a minimum cardinality PD-set of $G$ in $O(n^{8.5})$-time.
\end{theorem}
\begin{proof}
First, we show that Algorithm~\ref{algo} computes a minimum cardinality PD-set of an AT-free graph $G$. For any tuple $(X,X', size(X'))$ in queue $Q_i$, the set $X'$ represents a subsolution, $X$ represent the vertices picked in $X'$ from $(i-1)^{th}$ and $i^{th}$ levels, and $size(X')$ represents the cardinality of $X'$. We claim that for any  tuple $(X,X', size(X'))$ in queue $Q_i$ where $i \in [l]$,  $\bigcup_{j=0}^{i-1}L_j \subseteq N[X']$, that is, $X'$ dominates the all vertices $w$ such that $w \in L_r$ where $r \in [i-1]$ and if a vertex $w$ in $G[X']$ is unmatched then $w \in X' \cap L_i$,  that is, all vertices $w \in X' \setminus L_i$ is paired in $G[X']$. Note that this is true for $i=1$ because $Q_1$ contains all tuples $(X,X',size(X'))$ such that $X=X' \subseteq N[x]$ implying that $\{x\} = L_0 \subseteq N[X']$ satisfying the other property also.

Suppose the claim is true for $i-1$ where $i \in [l]$. Now we need to show that the claim is true for ordered tuples $(X,X', size(X'))$ in $Q_i$. Indeed, the tuple $(X,X', size(X'))$ is inserted in $Q_i$ only if there is tuple $(Y,Y', size(Y'))$ in $Q_{i-1}$, a set $A \subseteq L_i$ where, $ \vert  Y \cup A \vert  \leq 6$ such that, $L_{i-1} \subseteq N[Y \cup A]$, and there is a subset $B$ of $A$ such that induced subgraph $G[Y' \cup B]$ has a perfect matching. Note that $\cup_{j=0}^{i-1}L_j \subseteq N[Y']$ and all vertices $w \in Y' \setminus L_{i-2}$ is paired in $G[Y']$. Hence, by the way we have selected the set $A \subseteq L_{i+1}$, we have $\bigcup_{j=0}^{i-1}L_j \subseteq N[X']$ and all vertices $w \in X' \setminus L_{i-1}$ is paired in $G[X']$. This proves the claim.

Therefore, for any tuple $(X,X', size(X'))$ in $Q_l$ with $L_l \subseteq N[X]$ and $G[X']$ has a perfect matching, $X'$ is a PD-set of $G$. Hence, for any minimum cardinality PD-set $D_p$ of $G$ where $D_{p}$ contains at most $6$ vertices from any three consecutive BFS-levels of $x$, there will be a tuple $(X = D_p \cap (L_{l-1} \cup L_l)), X' = D_p, size(X') = \vert D_p\vert $ in $Q_l$ such that $L_l \subseteq N[X]$ and $G[X']$ has a perfect matching when Algorithm~\ref{algo} explores all BFS-levels of $x$. Consequently, we conclude that the Algorithm~\ref{algo} outputs a minimum cardinality PD-set of an AT-free graph.

Next, we analyse the time complexity of the Algorithm~\ref{algo}. Note that for each set $X$ and $U$ with $\vert  X \cup U \vert  \leq 6$, at most $O(n)$-time is required to check whether all the vertices at level $L_{i}$ are dominated. In addition for each $U' \subseteq U$ at most $O(m\sqrt{n})$ is required to check whether the subgraph $G[X' \cup U']$ has a perfect matching. Since there at most $2^6=64$ possibilities for $U'$, we need $O(m\sqrt{n})$ time to check check whether the subgraph $G[X' \cup U']$ has a perfect matching for all possibilities $U' \subseteq U$. Also note that there at most $O(n^6)$ subsets of $V$, whose size is at most $6$. Hence, the running time of the algorithm is $O(n^{8.5})$.
\end{proof}

\section{Paired Domination in Planar Graphs}
\label{sec5}
In this section we show that the \textsc{Decide PD-set} problem is NP-complete even when restricted to  planar graph. For this purpose, we will give a polynomial reduction from the  \textsc{Minimum Vertex Cover(Min-VC)} problem to  the \textsc{Min-PD} problem. In a graph $G=(V, E)$, a \emph{vertex cover} is a set $C \subseteq V$ such that $C$ has at  least one end point of every edge $e \in E$. The \textsc{Min-VC} problem require to compute a minimum cardinality vertex cover of a given graph $G$. The following theorem is already proved for the the \textsc{Min-VC} problem.

\begin{theorem}\cite{vc}
\label{vc}
The \textsc{Min-VC} problem is NP-hard for the planar cubic graphs.
\end{theorem}

Now, we prove the main result of this section.

\begin{theorem}
\label{planar}
The \textsc{Decide PD-set} problem is NP-complete for  planar graphs with maximum degree $5$.
\end{theorem}

\begin{proof}
Clearly, the \textsc{Decide PD-set} problem is in NP. To show the hardness of the problem, we give a reduction from \textsc{Min-VC} problem which is  NP-hard for planar cubic graphs, by Theorem~\ref{vc}. Let $G=(V,E)$ be a planar cubic graph with $V=\{v_1,v_2, \ldots, v_n\}$. We transform the graph $G$ into a  graph $G'=(V',E')$ as follows:
\begin{itemize}
\item replace each vertex $v_i \in V$ with the gadget $G_{v_i}$ as shown in the Fig.~\ref{p:planar}
 \item If three edges $e_{j},e_{k},e_{l}$ were incident on $v_{i}$ in $G$, then in $G'$, we make $e_{j}$ incident on $v_{i}^{1}$, $e_{k}$ incident on $v_{i}^{2}$ and $e_{l}$ incident on $v_{i}^{3}$.
\end{itemize}

\begin{figure}[ht]
\begin{center}
\includegraphics[width=0.8\textwidth]{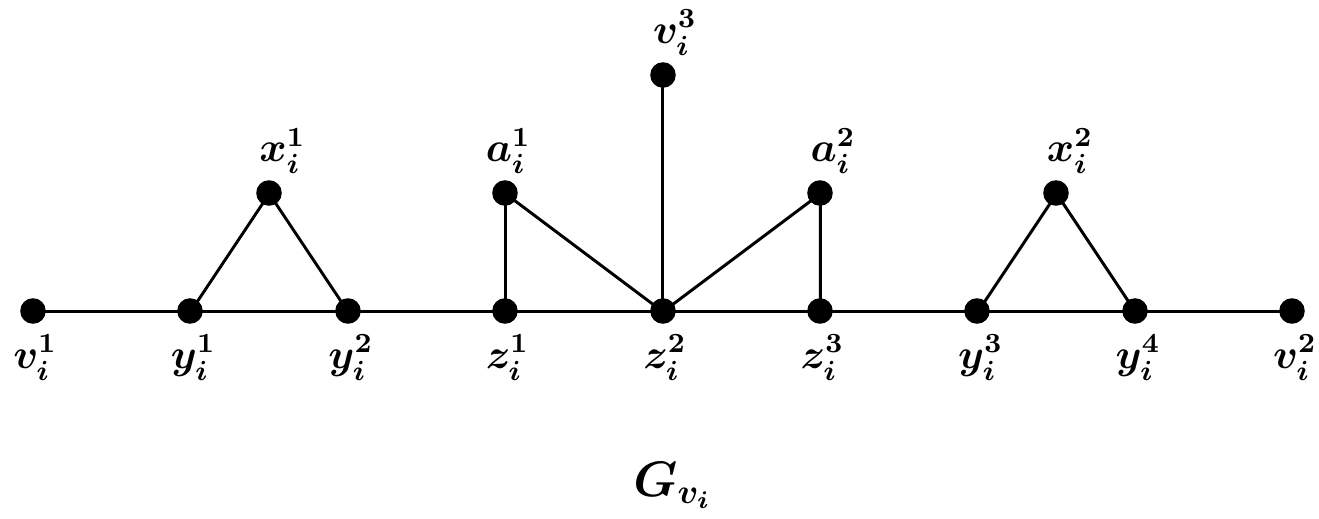}
\caption{Gadget $G_{v_i}$ used in the construction of graph $G'$ from $G$ in Theorem~\ref{planar}.}
\label{p:planar}
\end{center}
\end{figure}

We note that the graph $G'$ is a planar graph with maximum degree $5$, and $G'$ can be computed from $G$ in polynomial time. Now, to prove the result we only need to prove the following claim:

\begin{claim}
If $\beta(G)$ denotes the cardinality of a minimum vertex cover of $G$, then $\gamma_{pr}(G') = 4n + 2\beta(G)$, where $n$ denotes the number of vertices in $G$.
\end{claim}
\begin{proof}
Let $V^c$ be a minimum cardinality vertex cover of $G$. Let $D_p = \{v_{i}^{1}, y_{i}^{1}, v_{i}^{2}, y_{i}^{4},$ $v_{i}^{3}, z_{i}^{2} \mid v_i \in V^c\} \cup  \{y_{i}^{2}, z_{i}^{1}, y_{i}^{3}, z_{i}^{3} \mid v_i \notin V^c\}$ where $i \in [n]$. Note that if $v_{i} \notin V^{c}$, then all the three vertices adjacent to $v_{i}$ in $G$ must be present in $V^{c}$. Using this fact, it can be easily verified that $D_p$ is a PD-set of $G'$, and  $|D_p| = 6 \cdot \beta(G) + 4 \cdot (n-\beta(G))= 4n+ 2\beta(G)$. Therefore, if $D_{p}^{*}$ is a minimum cardinality PD-set of $G'$ then $|D_{p}^{*}| \le 4n + 2\beta(G)$. Hence, we have
\begin{equation}
\label{Eq1}
\gamma_{pr}(G') \le 4n + 2\beta(G)
\end{equation}

Conversely, suppose $D_{p}$ is a minimum cardinality PD-set of $G'$. Then, to dominate the vertex $x_{i}^{1}$, $D_p \cap \{x_{i}^{1}, y_{i}^{1}, y_{i}^{2}\}$ must be non-empty. Further, a vertex $ u \in \{x_{i}^{1}, y_{i}^{1}, y_{i}^{2}\}\cap D_{p}$ can  only be paired with a vertex in the set $\{v_{i}^{1},x_{i}^{1}, y_{i}^{1}, y_{i}^{2}, z_{i}^{1}\} \setminus \{u\}$. Hence, $|D_{p} \cap \{v_{i}^{1},x_{i}^{1}, y_{i}^{1}, y_{i}^{2}, z_{i}^{1}\} | \ge 2$. Similarly, we have $|D_{p} \cap \{v_{i}^{3},x_{i}^{3}, y_{i}^{3}, y_{i}^{4}, z_{i}^{3}\} | \ge 2$. Therefore, for each $i \in [n]$, we have $|D_{p} \cap V(G_{v_i})| \ge 4$. Note that to dominate $x_{i}^{1}$, $D_p \cap \{x_{i}^{1}, y_{i}^{1}, y_{i}^{2}\} \neq \emptyset$. Further, to dominate $a_{i}^{1}$, $D_p \cap \{z_{i}^{1}, z_{i}^{2}, a_{i}^{1}\} \neq \emptyset$. Similarly, to dominate $x_{i}^{2}$ and $a_{i}^{2}$, $D_p \cap \{x_{i}^{2}, y_{i}^{3}, y_{i}^{4}\} \neq \emptyset$ and $D_p \cap \{z_{i}^{3}, z_{i}^{2}, a_{i}^{2}\} \neq \emptyset$ respectively. Therefore, we observe that, if $|D_{p} \cap V(G_{v_i})| = 4$, then $D_{p} \cap V(G_{v_i}) = \{y_{i}^{2},z_{i}^{1},y_{i}^{3},z_{i}^{3}\}$.

Now, we prove that we can update $D_{p}$ such that $D_{p}$ remains a minimum cardinality PD-set of $G'$ and  for each $i \in [n]$, $|D_{p} \cap V(G_{v_i})| = 4$ or $|D_{p} \cap V(G_{v_i})| \ge 6$.  Suppose $|D_{p} \cap V(G_{v_i})| = 5$ for some $i \in [n]$. As we observed, the vertices dominating $x_{i}^{1}$ and $x_{i}^{2}$ are paired with the vertices of $V(G_{v_i})$, and $|D_{p} \cap V(G_{v_i})| \ge 4$. Hence if $|D_{p} \cap V(G_{v_i})| = 5$ then $D_{p} \cap \{v_{i}^{1}, v_{i}^{2}, v_{i}^{3}\} \neq \emptyset$, as only these vertices of the gadget $G_{v_i}$ can be paired with a vertex of another gadget.

\noindent \textbf{Case 1:} Suppose $v_{i}^{1} \in D_{p}$.

In this case, first we show that $D_{p} \cap \{ v_{i}^{2}, v_{i}^{3}\} = \emptyset$. Note that $v_{i}^{1}$ is paired with a vertex of some other gadget, and  $v_{i}^{1}$ is not dominating $x_{i}^{1}$. Further, if $u$ is the vertex dominating vertex $x_{i}^{1}$ then $u$ can only be paired with a vertex in the set $\{x_{i}^{1}, y_{i}^{1}, y_{i}^{2}, z_{i}^{1}\} \setminus \{u\}$. Therefore, $|D_{p} \cap \{v_{i}^{1},x_{i}^{1}, y_{i}^{1}, y_{i}^{2}, z_{i}^{1}\} | \ge 3$. Also as $|D_{p} \cap \{v_{i}^{2},x_{i}^{2}, y_{i}^{3}, y_{i}^{4}, z_{i}^{3}\} | \ge 2$ and $|D_{p} \cap V(G_{v_i})| = 5$, we have $v_{i}^{3} \notin D_{p}$. Further, as $|D_{p} \cap \{v_{i}^{1},x_{i}^{1}, y_{i}^{1}, y_{i}^{2}, z_{i}^{1}\} | \ge 3$ therefore, $|D_{p} \cap \{v_{i}^{2},x_{i}^{2}, y_{i}^{3}, y_{i}^{4}, z_{i}^{3}\} | = 2$. Now, if $v_i^{2} \in D_{p}$ then $v_{i}^{2}$ is paired with $y_{i}^{4}$ this leaves the vertex $z_{i}^{3}$ undominated, a contradiction. Therefore, $v_{i}^{2} \notin D_{p}$. This concludes that $D_{p} \cap \{ v_{i}^{2}, v_{i}^{3}\} = \emptyset$.

Now, let $v_{i}^{1}$ is paired with a vertex $u$ of another gadget, say $G_{v_j}$ where $i \neq j$. Note that $u \in \{v_{j}^{1},v_{j}^{2}, v_{j}^{3}\}$. It is easy to observe that $|D_{p} \cap V(G_{v_j})| \ge 5$. Suppose $v_{i}^{1}$ is paired with $v_{j}^{1}$. Now if  $y_{j}^{1} \notin D_{p}$ then update $D_{p}$ as follows: $D_{p} = D_{p} \setminus V(G_{v_i}) \cup \{y_{i}^{2}, y_{i}^{3}, z_{i}^{1}, z_{i}^{3}, y_{j}^{1}\}$ and pair $v_{j}^{1}$ with $y_{j}^{1}$. Now, suppose that $y_{j}^{1}$ already belongs to $D_{p}$. Note that $y_{j}^{1}$ is paired with either $y_{j}^{2}$ or $x_{j}^{1}$. If both $y_{j}^{2}$ and $x_{j}^{1} \in D_{p}$ then $y_{j}^{1}$ must be paired with $x_{j}^{1}$. In this case, the set $D_{p}' = D_{p} \setminus (V(G_{v_i}) \cup \{x_{j}^{1}\}) \cup \{y_{i}^{2}, y_{i}^{3}, z_{i}^{1}, z_{i}^{3}, y_{j}^{1}\}$ where $v_{j}^{1}$ is paired with $y_{j}^{1}$ is a PD-set of $G'$ and $|D_{p}'| < |D_{p}|$, a contradiction. Therefore, in this case either $y_{j}^{2} \notin D_{p}$ or $x_{j}^{1} \notin D_{p}$. If $x_{j}^{1} \notin D_{p}$ then $y_{j}^{1}$ is paired with $y_{j}^{2}$ and in this case, we update $D_p$ as follows: $D_{p} = D_{p} \setminus V(G_{v_i}) \cup \{y_{i}^{2}, y_{i}^{3}, z_{i}^{1}, z_{i}^{3}, x_{j}^{1}\}$, pair $v_{j}^{1}$ with $y_{j}^{1}$ and $y_{j}^{2}$ with $x_{j}^{1}$. We can update $D_{p}$ in a similar way if $y_{j}^{2} \notin D_{p}$.

Similarly we can update $D_{p}$ if $v_{i}^{1}$ is paired with $v_{j}^{2}$. Now suppose $v_{i}^{1}$ is paired with $v_{j}^{3}$. If $z_{j}^{2} \notin D_{p}$ then update $D_{p}$ as follows: $D_{p} = D_{p} \setminus V(G_{v_i}) \cup \{y_{i}^{2}, y_{i}^{3}, z_{i}^{1}, z_{i}^{3}, z_{j}^{2}\}$ and pair $v_{j}^{3}$ with $z_{j}^{2}$. But, if $z_{j}^{2} \in D_{p}$, we may observe that it is possible to update $D_{p}$ by giving similar arguments as above with suitable modifications, such that $v_{j}^{3}$ is paired with $z_{j}^{2}$. After update in each case, we may note that $|D_{p} \cap V(G_{v_i})| =4$ and $|D_{p} \cap V(G_{v_j})| \ge 6$.

\smallskip
\noindent \textbf{Case 2:} Suppose $v_{i}^{2} \in D_{p}$.

The arguments are similar to Case $1$.

\smallskip
\noindent \textbf{Case 3:} Suppose $v_{i}^{3} \in D_{p}$.

Let $v_{i}^{3}$ is paired with a vertex $u$ of another gadget $G_{v_j}$. Since, $|D_{p} \cap V(G_{v_i})| = 5$ we have $|D_{p} \cap \{v_{i}^{1},x_{i}^{1}, y_{i}^{1}, y_{i}^{2}, z_{i}^{1}\} | = 2$ and $|D_{p} \cap \{v_{i}^{3},x_{i}^{3}, y_{i}^{3}, y_{i}^{4}, z_{i}^{3}\} | = 2$. Now, if $v_i^{1} \in D_{p}$ then $v_{i}^{1}$ is paired with $y_{i}^{1}$ this leaves the vertex $z_{i}^{1}$ undominated, a contradiction. Similarly, if $v_i^{2} \in D_{p}$ then $v_{i}^{2}$ is paired with $y_{i}^{4}$ this leaves the vertex $z_{i}^{3}$ undominated, a contradiction. Hence, $D_{p} \cap \{v_{i}^{1}, v_{i}^{2}\} = \emptyset$.
Now we can give similar arguments as Case 1, to show that $D_{p}$ can be updated such that $|D_{p} \cap V(G_{v_i})| =4$ and $|D_{p} \cap V(G_{v_j})| \ge 6$.

Now, without loss of generality, we may assume that there exists a minimum cardinality PD-set of $G'$ such that for each $i \in [n]$, $|D_{p} \cap V(G_{v_i})| = 4$ or $|D_{p} \cap V(G_{v_i})|  \ge 6$

Define $V^{c} = \{v_i \in V \mid |D_{p} \cap V(G_{v_i})| \ge 6\}$. Next, we claim that $V^{c}$ is a vertex cover of $G$. Consider any two distinct vertices $v_i$ and $v_j$ in $G$ such that $v_iv_j \in E(G)$. We prove that either $|D_{p} \cap V(G_{v_i})| \ge 6$ or $|D_{p} \cap V(G_{v_j})| \ge 6$. Let $v_i^{k}$ is made adjacent  to $v_j^{k'}$, where $k,k' \in [3]$. Note that if $|D_{p} \cap V(G_{v_i})| = 4$ and $|D_{p} \cap V(G_{v_i})| = 4$ then from above observation, we have $D_{p} \cap V(G_{v_i}) = \{y_{i}^{2},z_{i}^{1},y_{i}^{3},z_{i}^{3}\}$ and $D_{p} \cap V(G_{v_j}) = \{y_{j}^{2},z_{j}^{1},y_{j}^{3},z_{j}^{3}\}$, this leaves the vertices $v_i^{k}$ and $v_j^{k'}$ undominated, a contradiction. Therefore, $V^{c}$ is a vertex cover of $G$. Also, $\gamma_{pr}(G')\geq 6|V^{c}|+4(n-|V^{c}|)$. So, we have $2|V^c| \le \gamma_{pr}(G')-4n$. Hence,
\begin{equation}
\label{Eq2}
2\beta(G) \le \gamma_{pr}(G')- 4n
\end{equation}

Therefore, using Equation~\ref{Eq1} and \ref{Eq2}, we have $\gamma_{pr}(G') = 4n + 2 \beta(G)$. This proves the claim.
\end{proof}

Since, the \textsc{Min-VC} problem is NP-hard for cubic planar graphs, from above claim we conclude that the \textsc{Decide PD-set} problem is NP-complete for planar graphs with maximum degree $5$.
\end{proof}

\section{Concluding Remarks}
\label{sec6}
In this paper, we resolve the complexity of the \textsc{Min-PD} problem for planar graphs and AT-free graphs. We proposed a polynomial time algorithm for \textsc{Min-PD} problem in AT-free graphs. We also proposed a $2$-approximation algorithm to compute a PD-set in AT-free graphs. Since the class of AT-free graphs include the class of cocomparability graphs, the results and algorithms presented for paired domination in AT-free graphs, also holds for cocomparability graphs. We further investigated the computational complexity of the problem in planar graphs and proved that the problem is NP-hard. The complexity of the problem is still not known in circle graphs. One may be interested in investigating the complexity status of the \textsc{Min-PD} problem in circle graph. Further, it is interesting to design more efficient algorithm for the problem in AT-free graphs and cocomparability graphs.


\end{document}